\documentclass[conference]{IEEEtran}
\IEEEoverridecommandlockouts
\usepackage[utf8]{inputenc} 
\usepackage[T1]{fontenc}
\usepackage{url}              
\usepackage{cite}             
\usepackage{amsmath,amssymb,amsfonts}
\usepackage{algorithmic}
\usepackage{graphicx}
\usepackage{textcomp}
\def\BibTeX{{\rm B\kern-.05em{\sc i\kern-.025em b}\kern-.08em
		T\kern-.1667em\lower.7ex\hbox{E}\kern-.125emX}}

\usepackage{mleftright}       
\mleftright                   
\usepackage{booktabs}         
\usepackage{verbatim}
\usepackage{color, soul}
\usepackage{epstopdf}
\usepackage{subfigure} 
\usepackage{algorithmic}
\usepackage{amsthm}
\usepackage{extarrows}
\usepackage{mathrsfs}
\newtheorem{theorem}{Theorem}

\usepackage{bbm}

\newcommand{\RNum}[1]{\uppercase\expandafter{\romannumeral #1\relax}}
\newtheorem{lemma}{Lemma}

\newtheorem{corollary}{Corollary}
\usepackage{stfloats}

\usepackage{color,soul}
\usepackage{tikz,xcolor,hyperref}

\definecolor{lime}{HTML}{A6CE39}
\allowdisplaybreaks[4]

\makeatletter

\newcommand*\bigcdot{\mathpalette\bigcdot@{.5}}
\newcommand*\bigcdot@[2]{\mathbin{\vcenter{\hbox{\scalebox{#2}{$\m@th#1\bullet$}}}}}
\makeatother

\usepackage[linesnumbered,ruled,vlined]{algorithm2e}

	\makeatletter
	\renewcommand{\maketag@@@}[1]{\hbox{\m@th\normalsize\normalfont#1}}%
	\makeatother
	
\begin{document}
		\title{Optimal Sampling for Uncertainty-of-Information Minimization in a Remote Monitoring System
		\vspace{-0.5em}}
		\author{%
			\IEEEauthorblockN{Xiaomeng~Chen\IEEEauthorrefmark{1}, Aimin~Li\IEEEauthorrefmark{1}, Shaohua~Wu\IEEEauthorrefmark{1}\IEEEauthorrefmark{2}}
			\IEEEauthorblockA{
				\IEEEauthorrefmark{1} Harbin Institute of Technology (Shenzhen), China \\
				\IEEEauthorrefmark{2} The Department of Boradband Communication, Peng Cheng Laboratory, Shenzhen, China \\
				23s052026@stu.hit.edu.cn, liaimin@stu.hit.edu.cn, hitwush@hit.edu.cn}
				\vspace{-3em}
			
		}				
		
		\maketitle
\begin{abstract}
	In this paper, we study a remote monitoring system where a receiver observes a remote binary Markov source and decides whether to sample and transmit the state through a randomly delayed channel. We adopt uncertainty of information (UoI), defined as the entropy conditional on past observations at the receiver, as a metric of value of information, in contrast to the traditional state-agnostic nonlinear age of information (AoI) penalty functions. To address the limitations of prior UoI research that assumes one-time-slot delays, we extend our analysis to scenarios with random delays. We model the problem as a partially observable Markov decision process (POMDP) problem and simplify it to a semi-Markov decision process (SMDP) by introducing the belief state. We propose two algorithms: A globally optimal \textit{bisection relative value iteration} (bisec-RVI) algorithm and a computationally efficient sub-optimal index-based threshold algorithm to solve the long-term average UoI minimization problem. Numerical simulations demonstrate that our sampling policies surpass traditional zero wait and AoI-optimal policies, particularly under conditions of large delay, with the sub-optimal policy nearly matching the performance of the optimal one.
\end{abstract}

\begin{IEEEkeywords}
	 Remote monitoring, uncertainty of information, age of information, Markov decision process
\end{IEEEkeywords}

\IEEEpeerreviewmaketitle

\section{Introduction} 
To evaluate the information freshness, age of information (AoI) has been proposed in \cite{AoI2012}, \cite{aoi2011}, and has attracted extensive research attention in remote monitoring, industrial automation, and internet-of-things (IoT) applications \cite{aoisurvey2021,DBLP:journals/tcom/DoganA21,wu2022minimizing,DBLP:journals/tmc/PanCLL23,xie2020age,meng2022analysis,DBLP:journals/iotj/WangWJWWZ22,DBLP:conf/globecom/Long0GLN22,DBLP:journals/tcom/FengWFCD24}.
Traditionally, AoI is known to be \textit{state-agnostic} \cite{AoII2020,pappas2021,newmaetrics2024}, \emph{i.e.}, focusing solely on the timeliness of information without accounting for the dynamics and the semantics of the source. Today, some variants of AoI have been proposed in \cite{10409276, voi2018,10323428,voi2022,sun2019non,SF2022,whittleforaoi2019,aoipenal2019,aoiopt2017,9551200} to address the \textit{state-agnostic} limitation. The researchers successfully derive their metrics of interest, such as the mutual information and the mean square error between the source and the latest receives message, as non-linear penalty functions of AoI. In this way, the problem of minimizing mutual information, or mean square error, can be transformed into a problem of minimizing non-linearly penalized AoI.

Uncertainty of information (UoI) is a new metric that addresses the \textit{state-agnostic} limitation of AoI \cite{UoI2022}. Defined as the entropy conditional on past observations at the receiver, UoI quantifies the receiver’s uncertainty about the latest state of the source based on previously received, potentially outdated observations. Unlike AoI, UoI integrates both the age and the historical observations, making it a \textit{state-aware} metric. For example, consider a binary discrete-time Markov process with state transition probabilities $P[0|1] = 0.2$ and $P[1|0] = 0.05$. As illustrated in Fig. \ref{figure11}, UoI's dependence on AoI varies with the latest received state $S_0$, exhibiting a non-monotonic relationship. This phenomenon introduces new challenges in designing UoI-optimal sampling policies.
\begin{figure}
	\centering
	\includegraphics[width=	0.8\linewidth]{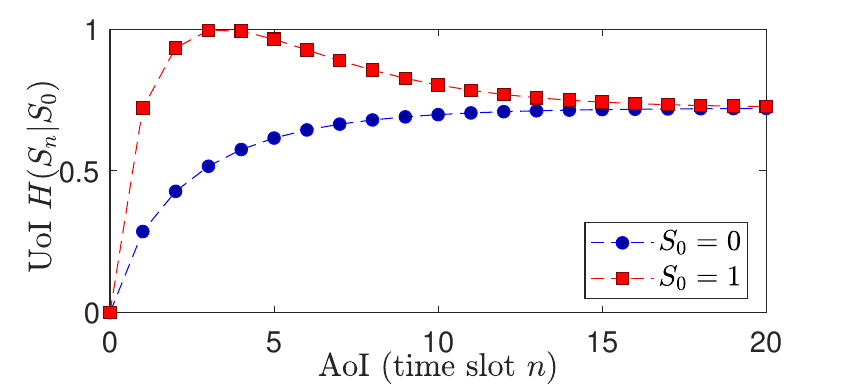}
	\caption{UoI vs. AoI and the latest observed state $S_0$.}
	\label{figure11} \vspace{-4mm}
\end{figure}

In \cite{UoI2022} and \cite{UoI2023}, the UoI-optimal sampling and scheduling policy has been investigated in the one-time-slot system. In \cite{UoI2022}, a Whittle index-based multi-source scheduling policy for binary Markov process is derived. Then an index policy for general finite-state Markov processes under unreliable channels are further extended in \cite{UoI2023}. Both of the studies idealize the transmission delay as one time slot. However, random and large delays are common in communication networks due to network loading, routing, and retransmission\cite{reasondelay2011}. Under random delay, UoI-optimal sampling process is no longer a typical Markov decision process (MDP), which brings about new challange. This prompts us to formulate a new problem for UoI-minimization under random delay. 

Up to this point, considerable research efforts have been devoted to optimizing metrics related to AoI in the presence of delays \cite{delay2020,twodelay2022,delay2022,twodelay2023,delay2024,tangAgeOptimalSampling2023}. Authors proposes a new “selection-from-buffer” model for sending the features aimed at minimizing the general functions of AoI (monotonic or non-monotonic) with random transmission delay in \cite{delay2022}. As an extension to \cite{delay2022}, authors minimizes the general functions of AoI through a channel with highly variable two-way random delay in \cite{delay2024}. And in \cite{tangAgeOptimalSampling2023}, an optimal sampling policy is designed to minimize the average AoI when the statistics of delay are unknown. However, all of these AoI-related functions to be minimized are state-agnostic, distinct to minimization of UoI, which is a state-aware metric.

To sum up, our motivation is twofold: ($i$) The optimization of UoI under random delay still remains largely unexplored. ($ii$) The UoI-minimization problem is distinct from other AoI-related optimization under random delay.
In this paper, we study a remote monitoring system where a receiver observes a remote binary Markov source and decides whether to sample the source's state over a randomly delayed channel. Our goal is to solve a UoI-minimization problem under random delay. We model the problem as a partially observed Markov decision process (POMDP), and simplify it into a semi-Markov decision process (SMDP) by introducing \emph{belief state}. Specifically, the contributions of this paper are as follows:
\begin{itemize}
	\item We formulate a new problem to minimize the average UoI under random delay. 
	\item We propose an optimal policy for this work. An optimal sampling policy is provided by applying a two layered \emph{bisection relative value iteration (bisec-RVI)} algorithm. 
	\item We develop a sub-optimal policy for this work. A sub-optimal index policy with computation efficiency is proposed based on the special properties of belief state.
\end{itemize}
Numerical simulations illustrate that our proposed sampling policies outperform traditional Zero wait and AoI-optimal policies. And the performance of the sub-optimal policy nearly match that of the optimal policy, especially under large delay.
		
\section{System Model and Problem Formulation} \label{section II}
\subsection{System Model}
  \begin{figure}
			\centering
			\includegraphics[width=	1\linewidth]{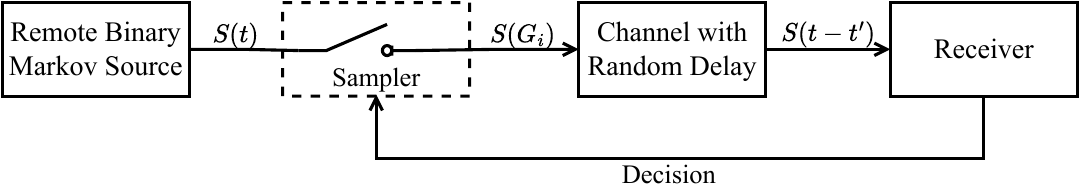}
			\caption{System model of the considered remote monitoring system.}
			\label{figure1} \vspace{-4mm}
		\end{figure}
We consider a discrete-time remote monitoring system in Fig. \ref{figure1}, where the states of a remote binary Markov source are delivered through a channel to a receiver. Based on the history observation of the source's state, the receiver decides whether to sample the current state or not, and transmits the decision to the sampler. 
The remote source is a discrete-time binary Markov process with a state at time $t$ denoted by $S(t) \in \{0,1\}$. We assume that the one-step transition matrix of the Markov process $\mathbf{P}$ is known to the receiver, given by
\begin{equation} \label{eq2-1}
  \begin{array}{cc}
  & \begin{matrix} \text{State 0} & \text{State 1} \end{matrix} \\
\begin{matrix} \text{State 0} \\ \text{State 1} \end{matrix} &
\begin{bmatrix}
  1-p & p\\
q & 1-q
\end{bmatrix}
\end{array}\end{equation}
		where $0 < p \leq q < 1$, $p+q \neq 1$ \footnote{Without loss of generality, we assume that $p \leq q$. And when $p+q = 1$, UoI is constant that can not be changed by any sampling policy.  For a similar reason, we also assume that $p+q \neq 0$ and $p+q \neq 2$.}.
		Consequently, the $n$-step transition matrix of this Markov process is \cite[Appendix A]{puterman2014markov}
		\begin{equation} \label{eq2-2}
			\mathbf{P}^n=
			\begin{bmatrix}
				1-p^{(n)} & p^{(n)}\\
				q^{(n)} & 1-q^{(n)}
			\end{bmatrix},
		\end{equation}
		where
		$p^{(n)} \triangleq \frac{p-p(1-p-q)^n}{p+q}$, $q^{(n)} \triangleq \frac{q-q(1-p-q)^n}{p+q}$, $n \in \mathbb{N^+}$.
		
		The $i$-th state information sent over the channel is sampled and transmitted in time slot $G_i$, and is delivered at the receiver at time slot $R_i = G_i +Y_i$, where $Y_i \geq 1$ is the \textit{independent and identically distributed} (i.i.d) random transmission delay of the $i$-th state information, satisfying $1 \leq E[Y_i] < \infty$. We assume that the sampler receives delay-free decision commands from the receiver \footnote{A remote monitoring system with two-way random delay can be extended based on this work.}
		, such that the receiver sends the decision and the sampler receives the decision at time slot $G_i$. Let $Z_i = G_i - R_{i-1} \geq 0$ represent the waiting time for sending the $i$-th packet after the ($i-1$)-th packet is received by the receiver.  
		
		
		\subsection{Problem Formmulation}
		The UoI is a metric proposed in \cite{UoI2022} to measure the uncertainty of the source at the receiver side given the history observations, given by
  \begin{equation}
      U(t)=H(S(t)|\mathbf{W}(t)),
  \end{equation}
  where $\mathbf{W}(t)=(S(G_0),S(G_1),\cdots,S(t'))$ are the history observations at the receiver side up to time slot $t$, and $t' \triangleq S(\mathop{\max}_{i} \{G_i:R_i \leq t\})$ is the time stamp of the most recently received update. Leveraging the Markov property of $S(t)$, we have 
		\begin{equation} \label{eq2-3}
  \begin{aligned}
      U(t) &\triangleq -\sum_{i\in\{0,1\}} \mathrm{P}[S(t)=i|\mathbf{W}(t)] \log_2 \mathrm{P}[S(t)=i|\mathbf{W}(t)]\\
      &=-\sum_{i\in\{0,1\}} \mathrm{P}[S(t)=i|S(t')] \log_2 \mathrm{P}[S(t)=i|S(t')],
  \end{aligned}
		\end{equation}
		From the definition of $t'$, we can rewrite $U(t)$ as a piecewise function:
		\begin{equation} \label{eq2-4}
			\begin{aligned}
				U(t) = -\sum_{i\in\{0,1\}} P[S(t)=i|S(G_i)] \log_2 P[S(t)=i|S(G_i)],\\
				\qquad \text{if}\,\, R_i \leq t < R_{i+1},\, \forall {i \in \mathbb{N}}.
			\end{aligned}
		\end{equation}
For short-hand notations, we introduce $H(p)\triangleq -p \log_2(p) - (1-p) \log_2 (1-p)$ as the entropy of a binary information. With \eqref{eq2-2} in hand, \eqref{eq2-4} can be rewritten as

\begin{align} U(t)=\left\{\begin{aligned} \label{eq2-6}
	&H(p^{(t-G_i)}), \,\, \text{if}\, S(G_i)=0, R_i \leq t < R_{i+1},\forall {i \in \mathbb{N}}\\
	&H(q^{(t-G_i)}), \,\, \text{if}\, S(G_i)=1, R_i \leq t < R_{i+1}, \forall {i \in \mathbb{N}},
\end{aligned}\right.\end{align}
where $t-G_i$ for $R_i \leq t < R_{i+1},\forall {i \in \mathbb{N}}$ is what exactly AoI indicates. As a result, UoI is an observation-based non-monotonic function of AoI.
		
A sampling policy is a sequence $\pi = \left(G_1,G_2,\cdots\right) \in \Pi$, which is time stamps of sampling times for each packet. Alternatively, the sampling policy can also be expressed as $\pi = \left(Z_1,Z_2,\cdots\right) \in \Pi$, which is the sequence of waiting times for each packet. (see \ref{A} for a detailed explanation.) Our goal is to find both an optimal sampling policy that minimizes the time-average expected sum-UoI:
\begin{equation} \label{eq2-7}
	\mathop{\inf}_{\pi \in \Pi} \mathop{\limsup}_{T \rightarrow \infty} \frac{1}{T} \mathbb{E} \left[\sum_{t=0}^{T-1}U(t)\right].
\end{equation}
		
This is a new problem different from previous studies focused on optimizing cost functions of AoI with random delay. The uniqueness resides in the fact that UoI is not only dependent on the age of the latest observation (\emph{i.e.} $t-G_i$), but also customized by the latest observation $S(G_i)$, as shown in \eqref{eq2-6}. Compared to UoI, the current functions of AoI used as metrics to design sampling policy under random delay, as we know, have nothing to do with the contents of the transmission information and remains invariant \cite{delay2022,twodelay2023,delay2024}.

		
\section{Optimal Sampling policy}\label{section III}	
In this section, we propose an optimal sampling policy by using \emph{bisec-RVI} algorithm to minimize the long-term average expected sum-UoI.
\subsection{Belief State}
In our system model, the receiver is tasked with determining sampling actions based on delayed and imperfect observations $\mathbf{W}(t)$ to minimize the time-average UoI. This problem is commonly modeled as a POMDP. A fundamental approach in solving partially observable Markov decision processes (POMDPs) involves transforming them into Markov Decision Processes (MDPs) by utilizing a concept known as \textit{belief state}. In this subsection, we explore the \textit{belief state} in our context, which is as the probability of $S(t)=1$ given the observations $\mathbf{W}(t)$, given by
\begin{equation}
    \Omega(t) = \mathrm{P} \left[ S(t)=1|\mathbf{W}(t)\right].
\end{equation}
Similar to the process to obtain \eqref{eq2-6}, we can prove that for $R_i \leq t < R_{i+1}$, $\forall i \in \mathbb{N}$, $\Omega (t)$ can be expressed by
\begin{align} \Omega(t)=\left\{\begin{aligned} \label{eq3-2}
		&p^{(t-G_i)}, \quad \text{if}\,\, S(G_i)=0\\
		&1-q^{(t-G_i)}, \quad \text{if}\,\, S(G_i)=1.
	\end{aligned}\right.\end{align}
The evolution of the \textit{belief state} is given in the following lemma:
\begin{lemma} \label{lemmabs}
    Given $\Omega(t)=\omega$, $\Omega(t+k)$ can be explicitly calculated by
    \begin{equation} \label{eq3-3}
	\Omega(t+k) = \frac{p-p(1-p-q)^k}{p+q} + \omega(1-p-q)^k,
\end{equation}
where $\omega \in \{ p^{(n)},\, 1-q^{(n)} \}$, $n,k \in \mathbb{N}$. For short-hand notations, we leverage $\tau^k(\omega)$ to denote the right-hand side of \eqref{eq3-3}. 
\end{lemma}
\begin{proof}
    Please refer to \ref{AA}.
\end{proof}
\begin{corollary}
    The \textit{equilibrium belief state} of $\Omega(t)$ is
    \begin{equation} \label{eq3-4}
	\omega^*\triangleq\mathop{\lim}_{k \rightarrow \infty} \tau^k(\omega) = \frac{p}{p+q}.
\end{equation}
\end{corollary}
\begin{proof}
    Since $0<|1-p-q|<1$, we have that $\mathop{\lim}_{k \rightarrow \infty} (1-p+q)^k=0$, and thus we obtain the limit.
\end{proof}

Since $H(p)=H(1-p)$, by combing \eqref{eq2-6} and \eqref{eq3-2} we have
\begin{equation}
    U(t)=H(\Omega(t)).
\end{equation}
Then the problem \eqref{eq2-7} can be rewritten as:
\vspace{-1mm}
\begin{equation} \label{eq3-5}
	\overline{p}_{\mathrm{opt}} = \mathop{\inf}_{\pi \in \Pi} \mathop{\limsup}_{T \rightarrow \infty} \frac{1}{T} \mathbb{E} \left[\sum_{t=0}^{T-1} H(\Omega(t)) \right],
\end{equation}
where $\overline{p}_{\mathrm{opt}}$ is the optimum value of \eqref{eq2-7}. 
\subsection{An Optimal Sampling Policy}
We present an optimal policy for \eqref{eq3-5} as follows:
\begin{theorem} \label{theorem1}
If $Y_i$'s are $i.i.d.$ with a finite mean $\mathbb{E}[Y_i]$, given $\Omega(R_i) = \omega$, then $\pi^{*} = (Z_1^{*}, Z_2^{*},\cdots)$ is an optimal solution to \eqref{eq3-5}, which satisfies the following optimality equation:
\begin{equation} \label{eqtheo1-1}
	\begin{aligned}
		&g(\bar{p}_{\mathrm{opt}}) + \widetilde{V}(\omega,\bar{p}_{\mathrm{opt}}) \\
		&= \mathop{\inf}_{Z_{i+1} \in \mathbb{N}}  \left\{c(\omega,Z_{i+1},\bar{p}_{\mathrm{opt}}) + r(\omega,Z_{i+1},\bar{p}_{\mathrm{opt}})\right\},
	\end{aligned}
\end{equation}
\begin{equation} \label{eqtheo1-2}
	g(\bar{p}_{\mathrm{opt}}) = \mathop{\inf}_{Z_{i+1} \in \mathbb{N}}  \{c(p,Z_{i+1},\bar{p}_{\mathrm{opt}}) + r(p,Z_{i+1},\bar{p}_{\mathrm{opt}})\},
\end{equation}
where 
\begin{equation}
	c(\omega,Z_{i+1},\bar{p}_{\mathrm{opt}}) =  \mathbb{E}\left[\sum_{k=0}^{Z_{i+1} +Y_{i+1}-1} (H(\tau^{k}(\omega))-\bar{p}_{\mathrm{opt}})\right],
\end{equation}
\begin{equation}
	\begin{aligned}
		r(\omega,Z_{i+1},\bar{p}_{\mathrm{opt}})=&\mathbb{E}\Bigg[\tau^{Z_{i+1}+Y_{i+1}}(\omega) \widetilde{V}(1-q^{(Y_{i+1})},\bar{p}_{\mathrm{opt}})\\
		&+(1-\tau^{Z_{i+1}+Y_{i+1}}(\omega))\widetilde{V}(p^{(Y_{i+1})},\bar{p}_{\mathrm{opt}})\Bigg] ,
	\end{aligned}
\end{equation}
for all $\omega \in \{ p^{(n)},\, 1-q^{(n)} \}$, $n \in \mathbb{N}$.
\end{theorem}
\begin{proof}[Proof sketch]
    The problem \eqref{eq3-5} can be cast as an infinite-horizon average cost SMDP \cite[Chapter 11]{puterman2014markov}. Recall that $Z_{i+1} = G_{i+1} - R_i$ as the waiting time for sending the $(i + 1)$-th packet after the $i$-th packet is received by the receiver. Given $\Omega(R_i) = \omega$, the Bellman optimality equation of the average cost problem is 
    \begin{equation} \label{eqtheorem1}
    	\begin{aligned}
    		V^{*}(\omega) = &\mathop{\inf}_{Z_{i+1} \in \mathbb{N}} \Biggl\{ \mathbb{E}\left[\sum_{k=0}^{Z_{i+1} +Y_{i+1}-1} (H(\tau^{k}(\omega))-\bar{p}_{\mathrm{opt}})\right] \\
    		&+  \mathbb{E}\Bigg[\tau^{Z_{i+1}+Y_{i+1}}(\omega)V^{*}(1-q^{(Y_{i+1})})\\
    		&+(1-\tau^{Z_{i+1}+Y_{i+1}}(\omega))V^{*}(p^{(Y_{i+1})})\Bigg] \Biggr\},
    	\end{aligned}
    \end{equation}
    for all $\omega \in \{ p^{(n)},\, 1-q^{(n)} \}$, $n \in \mathbb{N}$, where $V^*(\omega)$ is the relative value function associated with the average cost problem \eqref{eq3-5}. 
    
    We assume that $Y_{i}$'s are random variables with limited values, thus the states of $\omega$ are finite and countable. The equation \eqref{eqtheorem1} can be converted to \eqref{eqtheo1-1} and \eqref{eqtheo1-2}. This inversion can be interpreted as a transformation from the SMDP to an equivalent MDP \cite[Chapter 11]{puterman2014markov}.
    Please refer to Appendix \ref{A} for details of the proof.
\end{proof}

Applying Dinkelbach’s method for nonlinear fractional programming as shown in \cite{dinkelbach1967nonlinear} and \cite[lemma 2]{DBLP:journals/tit/SunPU20}, we get two assertions: ($i$) $\bar{p}_{\mathrm{opt}} \lesseqgtr\beta$ if and only if $g(\beta) \lesseqgtr 0$. ($ii$) $\bar{p}_{\mathrm{opt}}$ is the unique root of $g(\beta)=0$. Consequently, we can find the value of $\bar{p}_{\mathrm{opt}}$ by finding the root of $g(\beta)=0$ as shown in Algorithm \ref{Algorithm 1}, which is a two-layered \textit{Bisec-RVI} algorithm. Bisection search method is applied to the outer layer to get a fixed $\beta$ for every step and finally get the optimal value $\bar{p}_{\mathrm{opt}}$. In the inner layer, as the value of $\beta$ has been fixed by the outer layer, we only need to use RVI to find convergent $g(\beta)$. The RVI algorithm here is the same as it in MDP. The details about RVI algorithm have been neglected since it has been a mature technique to solve an infinite-horizon MDP \cite[Section 8.5.5]{puterman2014markov}. Similar algorithms are proposed in \cite{aoiopt2017,Sun2019nlage} and \cite{Tang2022wiener} to achieve age-optimal or mean square error (MSE)-optimal sampling.
\vspace{-2mm}
\begin{algorithm}
	\caption{Bisec-RVI algorithm}
	\label{Algorithm 1}
	\LinesNumbered
	\KwIn{$l=0$, $u=1$, tolerance $\epsilon >0$}
	\While {$u-l\ge \epsilon$}
	{
		$\beta:=(l+u)/2$\;
		Run RVI to solve $g(\beta)$ and $\widetilde{V}(\omega,\beta)$\;
		\If{$g(\beta)>0$}
		{
			$l := \beta$\;
		}
		\Else
		{$u:=\beta$\;}			
	}
	\KwOut{$\bar{p}_{\mathrm{opt}}=\beta$}
\end{algorithm}
\vspace{-3mm}

\section{Sub-optimal Index-based Policy}
The optimal policy for \eqref{eq3-5} uses \emph{bisec-RVI} algrithm, which requires repeatedly executing the RVI algorithm
in the inner layer until the value to be found converges, resulting in high computing complexity. 

In this section, we explore a sub-optimal but computation-efficient index-based policy in the sequel. We assume that the transmission process spends a long time, \emph{i.e.}, the value of $\mathbb{E}[Y_i]$, $\forall i \in \mathbb{N}$ is large enough. The large transmission delay may be due to long distance, time-varying channel conditions, too many packets in the channel and so on. According to \eqref{eq3-4}, we can conclude that the transition probability from receiving time $R_i$ to $R_{i+1}$ ($i \in \mathbb{N}$) is
\begin{equation} \label{eq3-6}
	\mathrm{P} [\Omega(R_{i+1}) = 1-q^{(Y_{i+1})} | \Omega(R_i) ] = \tau^{Z_{i+1}+Y_{i+1}}(\omega) = \omega^{*}.
\end{equation}
As the state transition probabilities are constant, the second and third terms of the right-hand-side of \eqref{eqtheorem1} are irrelevant to the waiting time, making the Bellman optimality equation easier to solve. Along this line, we can finally get an index-based sampling policy.

We depict the details of the problem as follows. On the basis of the assumption, we consider a sampling policy $\psi = (Z_1,Z_2,\cdots) \in \Pi$ and try to optimize the following problem:

\begin{equation} \label{eq3-7}
	\overline{p}_{\mathrm{nopt}} = \mathop{\inf}_{\psi \in \Pi} \mathop{\limsup}_{T \rightarrow \infty} \frac{1}{T} \mathbb{E} \left[\sum_{t=0}^{T-1} H(\Omega(t)) \right],
\end{equation}

After that, we present an index-based sampling policy for the problem. We introduce an index function as
\begin{equation} \label{eq3-8}
	\eta (\omega) \triangleq \inf_{Z_{i} \in \mathbb{N}^{+}} \frac{1}{Z_i} \sum_{k=0}^{Z_i -1} \mathbb{E} \left[H(\tau^{k+Y_i}(\omega) )\right],
\end{equation}
where 
$\omega \in \{ p^{(n)},\, 1-q^{(n)} \}$, $n \in \mathbb{N}$. 
		
\begin{theorem} \label{theorem2}
If $Y_i$'s are $i.i.d.$ with a finite mean $\mathbb{E}[Y_i]$, then $\psi = (Z_1(\beta_{\psi}), Z_2(\beta_\psi),\cdots)$ is an optimal solution to (15), where 
\begin{equation} \label{eq3-9}
	Z_{i+1}(\beta_{\psi}) = \min \{k \in \mathbb{N}:\eta(\Omega(t+k)) \geq \beta_{\psi},t \geq R_i(\beta_{\psi}) \},
\end{equation}
and $\beta_{\psi}$ is the unique root of
\begin{equation} \label{eq3-10}
	\mathbb{E} \left[\sum_{R_i(\beta_{\psi})}^{R_{i+1}(\beta_{\psi})-1} H(\Omega(t))\right] - \beta_{\psi} \mathbb{E} \left[ R_{i+1}(\beta_{\psi})-R_i(\beta_{\psi}) \right] = 0,
\end{equation}
where $R_i(\beta_{\psi}) = G_i(\beta_{\psi}) + Y_i$ is the receiving time of the $i$-th state information submitted to the channel, and $\Omega(t)$ is the belief state at time slot $t$. Moreover, $\beta_{\psi}$ is exactly the optimum value of (15), \emph{i.e.}, $\beta_{\psi} = \overline{p}_{\mathrm{nopt}}$.
\end{theorem}

\begin{proof}[Proof sketch]
The problem \eqref{eq3-7} can be cast as an infinite-horizon average cost SMDP. Given $\Omega(R_i) = \omega$, the Bellman optimality equation of the average cost problem is
\begin{equation} \label{eq3-11}
	\begin{aligned}
		V_{\psi}(\omega) = &\mathop{\inf}_{Z_{i+1} \in \mathbb{N}} \Biggl\{ \mathbb{E}\left[\sum_{k=0}^{Z_{i+1} +Y_{i+1}-1} (H(\tau^{k}(\omega))-\bar{p}_{\mathrm{nopt}})\right] \Biggr\}\\
		&+  \mathbb{E}\Big[\omega^{*}V_{\psi}(1-q^{(Y_{i+1})}) + (1-\omega^{*})V_{\psi}(p^{(Y_{i+1})})\Big] ,
	\end{aligned}
\end{equation}
for all $\omega \in \{ p^{(n)},\, 1-q^{(n)} \}$, $n \in \mathbb{N}$, where $V_{\psi}(\omega)$ is the relative value function associated with the average cost problem \eqref{eq3-7}. Theorem \ref{theorem2} is proven by directly solving  \eqref{eq3-11}. The details are provided in \ref{B}.
\end{proof}

Theorem \ref{theorem2} signifies that the optimal solution to \eqref{eq3-7} is an index-based threshold policy, where the index function depends on the belief state. Specifically, the state of the source is submitted in time slot $t$ if two conditions are satisfied: ($i$) The channel is idle in time slot $t$, ($ii$) the index $\eta(\Omega(t))$ exceeds a threshold $\beta_{\psi}$ (\emph{i.e.}, $\eta(\Omega(t)) \geq \beta_{\psi}$). where the threshold $\beta_{\psi}$ is exactly the optimum value of \eqref{eq3-7}. 

For notational simplicity, we rewrite \eqref{eq3-5} as
\begin{equation}
	f(\beta_{\psi}) = f_1(\beta_{\psi}) - \beta_{\psi} f_2(\beta_{\psi}) = 0,
\end{equation}
where $f_1(\beta_{\psi}) =  \mathbb{E} \left[	\sum_{R_i(\beta_{\psi})}^{R_{i+1}(\beta_{\psi})-1} H(\omega(t))\right]$ and $f_2(\beta_{\psi}) = \mathbb{E} \left[ R_{i+1}(\beta_{\psi})-R_i(\beta_{\psi}) \right]$.
Then a low-complexity algorithm for finding the optimal objective value $\bar{p}_{\mathrm{nopt}}$ is provided as algorithm \ref{Algorithm 2}.
\begin{figure}
	\centering
	\includegraphics[width=	0.9\linewidth]{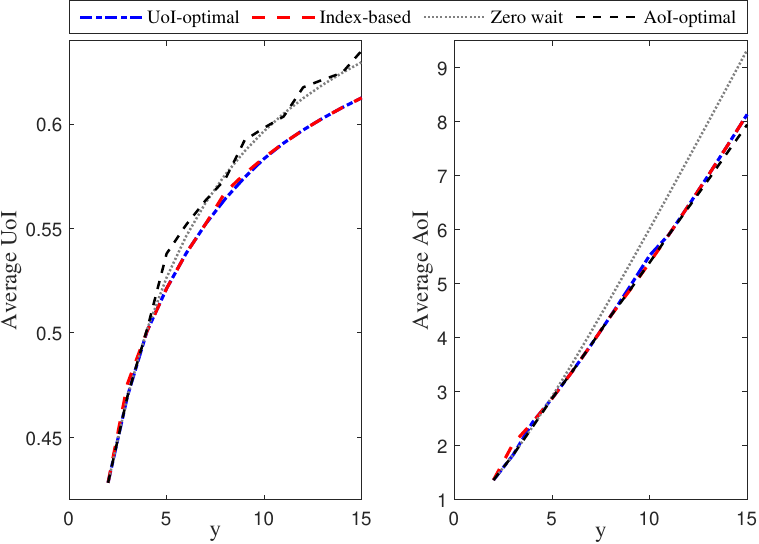}
	\caption{Average UoI and Average AoI v.s. $y$ with i.d.d random delay, where $P[Y_i = 1] = 0.8$ and $P[Y_i = y] = 0.2$, the dynamics of the Markov source depicted as $p = 0.05$ and $q = 0.2$. }
	\label{figure2}
\end{figure}

Compared the relative value function \eqref{eq3-11} with \eqref{eqtheorem1}, the core difference is that the state transition probabilities in \eqref{eq3-11} are irrelevant to the waiting time, which leads to an index policy for problem \eqref{eq3-7}. According to the convergence of $\tau^k(\omega)$ to $\omega^{*}$ in \eqref{eq3-4}, we induce that if $\mathbb{E}(Y_i)$ is large enough, the index policy is the optimal policy for \eqref{eq3-5}, \emph{i.e.}, $\pi^{*} = \psi$.

\begin{algorithm}
	\caption{Bisec-index algorithm}
	\label{Algorithm 2}
	\LinesNumbered
	\KwIn{$l=0$, $u=1$, tolerance $\epsilon >0$}
	\While {$u-l\ge \epsilon$}
	{
		$\beta_{\psi}:=(l+u)/2$\;
		$c := f(\beta_{\psi}) = f_1(\beta_{\psi}) - \beta_{\psi} f_2(\beta_{\psi})$\;
		\If{$c>0$}
		{
			$l := \beta_{\psi}$\;
		}
		\Else
		{$u:=\beta_{\psi}$\;}			
	}
	\KwOut{$\bar{p}_{\mathrm{nopt}}=\beta_{\psi}$} 
\end{algorithm}
		
\section{Numerical Results}\label{sectionIV}
This section presents numerical results that demonstrate the performance of the index-based threshold policy. As shown in \cite[Fig. 3]{UoI2022}, the $k$-step belief state evolutions with $p+q<1$ and with $p+q>1$ are quite different, so we operate both of them in the simulations. First, we evaluate the following four sampling policies:
\begin{enumerate}
	\item[1.] \emph{UoI-optimal:} The policy is given by Theorem \ref{theorem1}. 
	\item[2.] \emph{Index-based:} The policy is given by Theorem \ref{theorem2}.
	\item[3.] \emph{Zero wait:} An update is transmitted once the previous update is received, \emph{i.e.}, $Z_i=0$ for $\forall i \in \mathbb{N}$. This policy achieves the minimum delay and maximum throughput.
	\item[4.] \emph{AoI-optimal:} The AoI-optimal policy determines waiting time $Z_i$ by \cite[Theorem 4]{aoiopt2017} and \cite[Algorithm 2]{aoiopt2017}.
\end{enumerate}

\begin{figure}
	\centering
	\includegraphics[width=	0.9\linewidth]{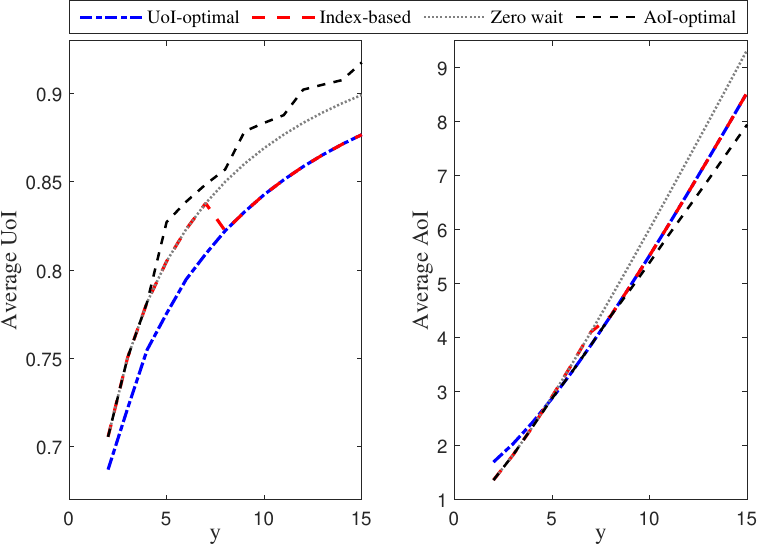}
	\caption{Average UoI and Average AoI v.s. $y$ with i.d.d random delay, where $P[Y_i = 1] = 0.8$ and $P[Y_i = y] = 0.2$, the dynamics of the Markov source depicted as $p = 0.7$ and $q = 0.95$.}
	\label{figure3} \vspace{-4mm}
\end{figure}

Fig. \ref{figure2} shows the four policies comparison in terms of average UoI and average AoI, when the binary Markov source evaluates with the probability $p+q<1$. The left panel shows that the average UoI obtained by the index policy is very close to the UoI-optimal one, compared to which Zero wait and AoI-optimal policy performs not well as $\mathbb{E}[Y_i]$ increases. However, the right panel reveals that AoI-optimal policy consistently achieves the lowest AoI. This implies that the desired goal the receiver tends to achieve leads to different result.

In Fig. \ref{figure3}, we compare the four policies performance for the evaluated case that $p+q>1$. The left panel shows a similar trend to the left panel of Fig. \ref{figure2}, except for the sub-optimal index-based policy. A watershed phenomenon occurs for the index-based policy: When $y \leq 6$, the index-based policy provides the same sampling strategy as zero wait policy, which performs worse than optimal policy; when $y > 6$, the index-based policy performs as well as the optimal policy. The reason for this phenomenon is that the index-based threshold policy is sub-optimal, deciding whether to sample or not by the comparison between the index and the constant value $\bar{p}_{\mathrm{nopt}}$. Therefore, this sub-optimal policy can not always take future circumstances into account if the value of $\mathbb{E}(Y_i)$ is small, thus ignores the oscillation of $H(\omega)$ when $p+q>1$. The index is always less than $\bar{p}_{\mathrm{nopt}}$ if $y \leq 6$, causing the zero wait policy. Otherwise, if $y > 6$, the index outweighs $\bar{p}_{\mathrm{nopt}}$ so the policy is optimal. The right panel demonstrates the lowest AoI is obtained by AoI-optimal policy as well.

The performance gains of UoI-optimal and index-based policies are close to the best average AoI, making them better choices when the system aims to optimize both AoI and UoI simultaneously. Moreover, Both of the pictures show that when $\mathbb{E}[Y_i]$ is large enough, the result of index policy is the same as that of the UoI-optimal policy, consistent with the theoretical induction we proposed before. But for what exact value $\mathbb{E}[Y_i]$ is meaning "large enough", is still an open issue.

\section{Conclusion}\label{sectionV}
In this paper, we have used UoI as a state-aware metric to estimate the value of information in a remote monitoring system. First we have put forward an optimal policy to minimize the time-average expected sum-UoI by two-layered \emph{bisec-RVI} algorithm. Based on the properties of belief state, we have further provided a sub-optimal index-based sampling policy owning lower computing complexity than the optimal one. The good performance of the sampling policies have been demonstrated by numerical simulations. Both of the proposed sampling policies outperform zero wait policy and AoI-optimal policy. Moreover, the performance of the sub-optimal policy approaches to that of the optimal policy, particularly under large delay. In the future work, it may be worthwhile to investigate the specific value of $\mathbb{E}[Y_i]$ that leads to the sub-optimal policy being identical to the optimal policy.

\bibliographystyle{IEEEtran}
\newpage
\bibliography{reference}

\begin{thebibliography}{10}
\providecommand{\url}[1]{#1}
\csname url@samestyle\endcsname
\providecommand{\newblock}{\relax}
\providecommand{\bibinfo}[2]{#2}
\providecommand{\BIBentrySTDinterwordspacing}{\spaceskip=0pt\relax}
\providecommand{\BIBentryALTinterwordstretchfactor}{4}
\providecommand{\BIBentryALTinterwordspacing}{\spaceskip=\fontdimen2\font plus
\BIBentryALTinterwordstretchfactor\fontdimen3\font minus
  \fontdimen4\font\relax}
\providecommand{\BIBforeignlanguage}[2]{{%
\expandafter\ifx\csname l@#1\endcsname\relax
\typeout{** WARNING: IEEEtran.bst: No hyphenation pattern has been}%
\typeout{** loaded for the language `#1'. Using the pattern for}%
\typeout{** the default language instead.}%
\else
\language=\csname l@#1\endcsname
\fi
#2}}
\providecommand{\BIBdecl}{\relax}
\BIBdecl

\bibitem{AoI2012}
S.~K. Kaul, R.~D. Yates, and M.~Gruteser, ``Real-time status: How often should
  one update?'' in \emph{Proc. {IEEE} {INFOCOM}}, Mar. 2012, pp. 2731--2735.

\bibitem{aoi2011}
S.~K. Kaul, M.~Gruteser, V.~Rai, and J.~B. Kenney, ``Minimizing age of
  information in vehicular networks,'' in \emph{Proc. 8th Annu. {IEEE} Commun.
  Soc. Conf. Sensor, Mesh Ad Hoc Commun. and Netw.}, Jun. 2011, pp. 350--358.

\bibitem{aoisurvey2021}
R.~D. Yates, Y.~Sun, D.~R. Brown, S.~K. Kaul, E.~Modiano, and S.~Ulukus, ``Age
  of information: An introduction and survey,'' \emph{{IEEE} J. Sel. Areas
  Commun.}, vol.~39, no.~5, pp. 1183--1210, Sep. 2021.

\bibitem{DBLP:journals/tcom/DoganA21}
O.~Dogan and N.~Akar, ``The multi-source probabilistically preemptive
  {M/PH/1/1} queue with packet errors,'' \emph{{IEEE} Trans. Commun.}, vol.~69,
  no.~11, pp. 7297--7308, Dec. 2021.

\bibitem{wu2022minimizing}
S.~Wu, Z.~Deng, A.~Li, J.~Jiao, N.~Zhang, and Q.~Zhang, ``Minimizing
  age-of-information in {HARQ-CC} aided {NOMA} systems,'' \emph{{IEEE} Trans.
  Wirel. Commun.}, vol.~22, no.~2, pp. 1072--1086, Feb. 2023.

\bibitem{DBLP:journals/tmc/PanCLL23}
H.~Pan, T.~Chan, V.~C.~M. Leung, and J.~Li, ``Age of information in
  physical-layer network coding enabled two-way relay networks,'' \emph{{IEEE}
  Trans. Mob. Comput.}, vol.~22, no.~8, pp. 4485--4499, Jul. 2023.

\bibitem{xie2020age}
M.~Xie, Q.~Wang, J.~Gong, and X.~Ma, ``Age and energy analysis for {LDPC} coded
  status update with and without {ARQ},'' \emph{{IEEE} Internet Things J.},
  vol.~7, no.~10, pp. 10\,388--10\,400, Oct. 2020.

\bibitem{meng2022analysis}
S.~Meng, S.~Wu, A.~Li, J.~Jiao, N.~Zhang, and Q.~Zhang, ``Analysis and
  optimization of the {HARQ}-based spinal coded timely status update system,''
  \emph{{IEEE} Trans. Commun.}, vol.~70, no.~10, pp. 6425--6440, Nov. 2022.

\bibitem{DBLP:journals/iotj/WangWJWWZ22}
Y.~Wang, S.~Wu, J.~Jiao, W.~Wu, Y.~Wang, and Q.~Zhang, ``Age-optimal
  transmission policy with {HARQ} for freshness-critical vehicular status
  updates in space-air-ground-integrated networks,'' \emph{{IEEE} Internet
  Things J.}, vol.~9, no.~8, pp. 5719--5729, Aug. 2022.

\bibitem{DBLP:conf/globecom/Long0GLN22}
Y.~Long, W.~Zhang, S.~Gong, X.~Luo, and D.~Niyato, ``{AoI}-aware scheduling and
  trajectory optimization for multi-{UAV}-assisted wireless networks,'' in
  \emph{{IEEE} {GLOBECOM}}, Jan. 2022, pp. 2163--2168.

\bibitem{DBLP:journals/tcom/FengWFCD24}
H.~Feng, J.~Wang, Z.~Fang, J.~Chen, and D.~Do, ``Evaluating aoi-centric {HARQ}
  protocols for {UAV} networks,'' \emph{{IEEE} Trans. Commun.}, vol.~72, no.~1,
  pp. 288--301, Feb. 2024.

\bibitem{AoII2020}
A.~Maatouk, S.~Kriouile, M.~Assaad, and A.~Ephremides, ``The age of incorrect
  information: {A} new performance metric for status updates,''
  \emph{{IEEE/ACM} Trans. Netw.}, vol.~28, no.~5, pp. 2215--2228, Nov. 2020.

\bibitem{pappas2021}
M.~Kountouris and N.~Pappas, ``Semantics-empowered communication for networked
  intelligent systems,'' \emph{{IEEE} Commun. Mag.}, vol.~59, no.~6, pp.
  96--102, Jul. 2021.

\bibitem{newmaetrics2024}
\BIBentryALTinterwordspacing
A.~Li, S.~Wu, S.~Meng, R.~Lu, S.~Sun, and Q.~Zhang, ``Towards goal-oriented
  semantic communications: New metrics, open challenges, and future research
  directions,'' \emph{IEEE Wirel. Commun. to appear}, 2024. [Online].
  Available: \url{https://doi.org/10.48550/arXiv.2304.00848}
\BIBentrySTDinterwordspacing

\bibitem{10409276}
M.~Salimnejad, M.~Kountouris, and N.~Pappas, ``Real-time reconstruction of
  {M}arkov sources and remote actuation over wireless channels,'' \emph{IEEE
  Trans. Commun.}, Jan. 2024.

\bibitem{voi2018}
Y.~Sun and B.~Cyr, ``Information aging through queues: {A} mutual information
  perspective,'' in \emph{Proc. {IEEE} {SPAWC} Workshop}, 2018.

\bibitem{10323428}
M.~Salimnejad, M.~Kountouris, and N.~Pappas, ``State-aware real-time tracking
  and remote reconstruction of a {M}arkov source,'' \emph{J. Commun. Networks},
  vol.~25, no.~5, pp. 657--669, Oct. 2023.

\bibitem{voi2022}
Z.~Wang, M.~Badiu, and J.~P. Coon, ``A framework for characterizing the value
  of information in hidden markov models,'' \emph{{IEEE} Trans. Inf. Theory},
  vol.~68, no.~8, pp. 5203--5216, Jul. 2022.

\bibitem{sun2019non}
Y.~Sun and B.~Cyr, ``Sampling for data freshness optimization: Non-linear age
  functions,'' \emph{J. Commun. Networks}, vol.~21, no.~3, pp. 204--219, 2019.

\bibitem{SF2022}
J.~P. Champati, M.~Skoglund, M.~Jansson, and J.~Gross, ``Detecting state
  transitions of a markov source: Sampling frequency and age trade-off,''
  \emph{{IEEE} Trans. Commun.}, vol.~70, no.~5, pp. 3081--3095, Jun. 2022.

\bibitem{whittleforaoi2019}
V.~Tripathi and E.~H. Modiano, ``A whittle index approach to minimizing
  functions of age of information,'' in \emph{Proc. 57th Annu. Allerton Conf.
  on Commun., Control, and Comput., (Allerton)}, 2019, pp. 1160--1167.

\bibitem{aoipenal2019}
M.~Kl{\"{u}}gel, M.~H. Mamduhi, S.~Hirche, and W.~Kellerer, ``{AoI}-penalty
  minimization for networked control systems with packet loss,'' in \emph{Proc.
  {IEEE} Conf. on Comput. Commun. Workshops}, 2019, pp. 189--196.

\bibitem{aoiopt2017}
Y.~Sun, E.~Uysal{-}Biyikoglu, R.~D. Yates, C.~E. Koksal, and N.~B. Shroff,
  ``Update or wait: {H}ow to keep your data fresh,'' \emph{{IEEE} Trans. Inf.
  Theory}, vol.~63, no.~11, pp. 7492--7508, Mar. 2017.

\bibitem{9551200}
N.~Pappas and M.~Kountouris, ``Goal-oriented communication for real-time
  tracking in autonomous systems,'' in \emph{Proc. {IEEE} Inter. Conf. on
  Autom. Syst.}, 2021, pp. 1--5.

\bibitem{UoI2022}
G.~Chen, S.~C. Liew, and Y.~Shao, ``Uncertainty-of-information scheduling: {A}
  restless multi-armed bandit framework,'' \emph{{IEEE} Trans. Inf. Theory},
  vol.~68, no.~9, pp. 6151--6173, Srp. 2022.

\bibitem{UoI2023}
G.~Chen and S.~C. Liew, ``An index policy for minimizing the
  uncertainty-of-information of markov sources,'' \emph{{IEEE} Trans. Inf.
  Theory}, vol.~70, no.~1, pp. 698--721, Jan. 2023.

\bibitem{reasondelay2011}
W.~Yao, L.~Jiang, Q.~H. Wu, J.~Y. Wen, and S.~J. Cheng, ``Delay-dependent
  stability analysis of the power system with a wide-area damping controller
  embedded,'' \emph{{IEEE} Trans. Power Syst.}, vol.~26, no.~1, pp. 233--240,
  Dec. 2011.

\bibitem{delay2020}
Y.~Sun, Y.~Polyanskiy, and E.~Uysal, ``Sampling of the {W}iener process for
  remote estimation over a channel with random delay,'' \emph{IEEE Trans. Inf.
  Theory}, vol.~66, no.~2, pp. 1118--1135, 2020.

\bibitem{twodelay2022}
C.~Tsai and C.~Wang, ``Unifying {AoI} minimization and remote estimation -
  optimal sensor/controller coordination with random two-way delay,''
  \emph{{IEEE/ACM} Trans. Netw.}, vol.~30, no.~1, pp. 229--242, Mar. 2022.

\bibitem{delay2022}
M.~K.~C. Shisher and Y.~Sun, ``How does data freshness affect real-time
  supervised learning?'' in \emph{Proc. {ACM} MobiHoc}, 2022, pp. 31--40.

\bibitem{twodelay2023}
M.~Moltafet, M.~Leinonen, M.~Codreanu, and R.~D. Yates, ``Status update control
  and analysis under two-way delay,'' \emph{{IEEE/ACM} Trans. Netw.}, vol.~31,
  no.~6, pp. 2918--2933, Jan. 2023.

\bibitem{delay2024}
\BIBentryALTinterwordspacing
C.~Ari, M.~K.~C. Shisher, E.~Uysal, and Y.~Sun, ``Goal-oriented communications
  for remote inference under two-way delay with memory,'' \emph{arXiv}, 2023.
  [Online]. Available: \url{https://doi.org/10.48550/arXiv.2311.11143}
\BIBentrySTDinterwordspacing

\bibitem{tangAgeOptimalSampling2023}
H.~Tang, Y.~Chen, J.~Wang, P.~Yang, and L.~Tassiulas, ``Age optimal sampling
  under unknown delay statistics,'' \emph{IEEE Trans. Inf. Theory}, vol.~69,
  no.~2, pp. 1295--1314, Feb. 2023.

\bibitem{puterman2014markov}
M.~L. Puterman, \emph{Markov Decision Processes: {D}iscrete Stochastic Dynamic
  Programming}.\hskip 1em plus 0.5em minus 0.4em\relax Hoboken, {NJ}, USA:
  Wiley, 2005.

\bibitem{dinkelbach1967nonlinear}
W.~Dinkelbach, ``On nonlinear fractional programming,'' \emph{Manag. Sci},
  vol.~13, no.~7, pp. 492--498, 1967.

\bibitem{DBLP:journals/tit/SunPU20}
Y.~Sun, Y.~Polyanskiy, and E.~Uysal, ``Sampling of the {W}iener process for
  remote estimation over a channel with random delay,'' \emph{{IEEE} Trans.
  Inf. Theory}, vol.~66, no.~2, pp. 1118--1135, 2020.

\bibitem{Sun2019nlage}
Y.~Sun and B.~Cyr, ``Sampling for data freshness optimization: {N}on-linear age
  functions,'' \emph{J. Commun. Networks}, vol.~21, no.~3, pp. 204--219, Sep.
  2019.

\bibitem{Tang2022wiener}
H.~Tang, Y.~Sun, and L.~Tassiulas, ``Sampling of the {W}iener process for
  remote estimation over a channel with unknown delay statistics,'' in
  \emph{Proc. MobiHoc}, 2022, p. 51–60.

\bibitem{dynamic2012}
D.~Bertsekas, \emph{Dynamic programming and optimal control: Volume I}.\hskip
  1em plus 0.5em minus 0.4em\relax Nashua, {NH}, USA: Athena Scientific, 2012.

\bibitem{krishnamurthy2016partially}
V.~Krishnamurthy, \emph{Partially observed Markov decision processes}.\hskip
  1em plus 0.5em minus 0.4em\relax Cambridge, U.K.: Cambridge University Press
  \& Assessment, 2016.

\bibitem{haas2006stochastic}
P.~Haas, \emph{Stochastic petri nets: {M}odelling, stability,
  simulation}.\hskip 1em plus 0.5em minus 0.4em\relax Berlin, Germany: Springer
  Science \& Business Media, 2006.

\end{thebibliography}
		
\clearpage
\appendices
\section{The Proof of lemma \ref{lemmabs}} \label{AA}
Here is the detailed process for the evaluation of the belief state. With the transition matrix given in \eqref{eq2-2}, the evolution of the belief state can be written as
\begin{equation}
	\begin{aligned}
		\Omega(t+k) &= \Omega(t) \cdot (1-q^{(k)}) + (1-\Omega(t)) \cdot p^{(k)} \\
		&= p^{(k)} + \Omega(t)(1-p^{(k)}-q^{(k)}) \\
		&= \frac{p-p(1-p-q)^k}{p+q} + \Omega(t)(1-p-q)^k \\
		&= \frac{p-p(1-p-q)^k}{p+q} + \omega(1-p-q)^k
	\end{aligned}
\end{equation}The prood is completed.

\section{The proof of Theorem \ref{theorem1}} \label{A}
The sampling problem \eqref{eq3-5} is an infinite-horizon average cost semi-Markov decision process (SMDP) \cite[Chapter 5.6]{dynamic2012} \cite[Chapter 11]{puterman2014markov}. In accordance with the description provided in \cite[Chapter 5.6]{dynamic2012}, we elaborate on the various components of this problem blow: 
\begin{itemize}
	\item \textbf{Decision Time}: Because the control commands from the receiver are promptly delivered to the sampler without delay, each $i$-th receiving time $R_i = G_i + Y_i$ also serves as a decision time of the problem \eqref{eq3-5}, where $G_i$ is the submission time of the $i$-th packet, taking $Y_i$ time slots to be delivered to the receiver.
	\item \textbf{State}: The belief state $\Omega(t)$. At time slot $R_i$, the state of the system is the belief state $\Omega(R_i)$.
	\item \textbf{Action}: $Z_{i+1} = G_{i+1} - R_{i}$ is the waiting time for submitting the $(i+1)$-th packet after the $i$-th packet is received. We consider $G_0 = 0$, $G_i = \sum_{j=0}^{i-1} (Y_j + Z_{j+1})$ and $R_i = G_i + Y_i$ for each $i \in \{ 0,1,\cdots \}$. Consequently, given $(Y_0,Y_1,\cdots)$, the sequence $(G_1,G_2,\cdots)$ is uniquely determined by $(Z_1,Z_2,\cdots)$. Therefore, $ \pi = (Z_1,Z_2,\cdots)$ is a valid representation to a policy in $\Pi$.
	\item  \textbf{State Transitions}: The state $\Omega(t)$ evolves as follows:
	\begin{align} \Omega(t)=\left\{\begin{aligned} 
			&1-q^{(Y_i)}, \quad\;\; \text{if}\,\, t=R_i,S(G_i)=1 \\
			&p^{(Y_i)}, \quad \quad\quad\: \text{if}\,\, t=R_i,S(G_i)=0\\
			&\tau(\Omega(t-1)), \,\, \text{otherwise},
		\end{aligned}\right.\end{align}
	where $i \in \mathbb{N}$. 
	From \eqref{eq3-3} we can get the transition probability is
	\begin{align} &P[\Omega(R_{i+1})|\Omega(R_{i})] \\
		&=\left\{\begin{aligned} 
			&\tau^{Z_{i+1}+Y_{i+1}}(\Omega(R_{i})),  \quad\quad\: \text{if}\,\, \Omega(R_{i+1}) = 1-q^{(Y_{i+1})}\\
			&1-\tau^{Z_{i+1}+Y_{i+1}}(\Omega(R_{i})), \;\; \text{if}\,\, \Omega(R_{i+1}) = p^{(Y_{i+1})}.
		\end{aligned}\right.\end{align}
	\item \textbf{Expected Transition Time}: The expected time difference between the decision times $R_i$ and $R_{i+1}$ is $\mathbb{E}[R_{i+1} - R_i] = \mathbb{E}[Z_{i+1} + Y_{i+1}]$
	\item \textbf{Expected Transition Cost}: The expected cumulative cost incurred during the interval between the decision times $R_i$ and $R_{i+1}$ can be calculated as
	\begin{equation}
		\begin{aligned}
			\mathbb{E} \left[ \sum_{t=R_i}^{R_{i+1}-1} H(\Omega(t)) \right] &= \mathbb{E} \left[ \sum_{k=0}^{Z_{i+1} + Y_{i+1}-1} H(\Omega(R_i+k)) \right] \\
			&= \mathbb{E}\left[\sum_{k=0}^{Z_{i+1} +Y_{i+1}-1} H(\tau^{k}(\Omega(R_i)))\right]
		\end{aligned}
	\end{equation}
\end{itemize}

We can use dynamic programming \cite{dynamic2012} to solve the average-cost SMDP \eqref{eq3-5}. If $\Omega(R_i) = \omega$, then the optimal decision $Z_{i+1}$ for the $(i+1)$-th packet satisfies the following Bellman optimality equation
\begin{equation} \label{b-5}
	\begin{aligned}
		V^{*}(\omega) = &\mathop{\inf}_{Z_{i+1} \in \mathbb{N}} \Biggl\{ \mathbb{E}\left[\sum_{k=0}^{Z_{i+1} +Y_{i+1}-1} (H(\tau^{k}(\omega))-\bar{p}_{\mathrm{opt}})\right] \\
		&+  \mathbb{E}\Bigg[\tau^{Z_{i+1}+Y_{i+1}}(\omega)V^*(1-q^{(Y_{i+1})})\\
		&+(1-\tau^{Z_{i+1}+Y_{i+1}}(\omega))V^*(p^{(Y_{i+1})})\Bigg] \Biggr\},
	\end{aligned}
\end{equation}
where $V^{*}(\omega)$ is the relative value function associated with the average cost problem \eqref{eq3-7}.

According to \cite[Chapter 11]{puterman2014markov}, we can transform the SMDP into an equivalent Markov decision process (MDP).
First, we rewrite the objective function in problem \eqref{eq3-5} by taking the limit superior of the ratio of the total expected transition cost to the $n$-th decision epoch ($n \in \mathbb{N^+}$) to the total expected transition time as follows:
\begin{equation} \label{b-6}
	\begin{aligned}
		&\mathop{\limsup}_{T \rightarrow \infty} \frac{1}{T} \mathbb{E} \left[\sum_{t=0}^{T-1} H(\Omega(t)) \right] \\
		&= \mathop{\lim}_{n \rightarrow \infty} \frac{\sum_{i=0}^{n-1} \mathbb{E} \left[ \sum_{t=R_i}^{R_{i+1}-1} H(\Omega(t)) \right] }{\sum_{i=0}^{n-1} \mathbb{E}[R_{i+1} - R_i] } \\
		&= \mathop{\lim}_{n \rightarrow \infty} \frac{\sum_{i=0}^{n-1} \mathbb{E}\left[\sum_{k=0}^{Z_{i+1} +Y_{i+1}-1} H(\tau^{k}(\omega))\right] }{\sum_{i=0}^{n-1} \mathbb{E}[Z_{i+1} + Y_{i+1}] },
	\end{aligned}
\end{equation}
where $\Omega(R_i) = \omega$.
As a result, \eqref{eq3-5} can be expressed as
\begin{equation} \label{b-7}
	\bar{p}_{\mathrm{opt}} = \mathop{\inf}_{\pi \in \Pi} \mathop{\lim}_{n \rightarrow \infty} \frac{\sum_{i=0}^{n-1} \mathbb{E}\left[\sum_{k=0}^{Z_{i+1} +Y_{i+1}-1} H(\tau^{k}(\omega))\right] }{\sum_{i=0}^{n-1} \mathbb{E}[Z_{i+1} + Y_{i+1}] }.
\end{equation}

To solve the problem \eqref{b-7}, we consider the following problem with a parameter $\beta \geq 0$:
	\begin{align}\label{b-8}
		&g(\beta) \triangleq 
		\mathop{\inf}_{\phi \in \Phi} \mathop{\lim}_{n \rightarrow \infty} \frac{1}{n} \sum_{i=0}^{n-1} \Bigg\{ \mathbb{E}\left[\sum_{k=0}^{Z_{i+1} +Y_{i+1}-1} H(\tau^{k}(\omega))\right] \notag \\ 
		& \quad\qquad\qquad\qquad\qquad\qquad\qquad -\beta  \mathbb{E}[Z_{i+1} + Y_{i+1}] \Bigg\} \notag \\
		&= \mathop{\inf}_{\phi \in \Phi} \mathop{\lim}_{n \rightarrow \infty} \frac{1}{n} \sum_{i=0}^{n-1}  \mathbb{E}\left[\sum_{k=0}^{Z_{i+1} +Y_{i+1}-1} [H(\tau^{k}(\omega))-\beta]\right] .
	\end{align}

For each constant $\beta$, the problem \eqref{b-8} can be cast as a MDP problem with components below:
\begin{itemize}
	\item \textbf{State}: The state of the MDP is the belief state $\Omega(R_i)$ at time $R_i$ for $\forall i \in \mathbb{N}$.
	\item \textbf{Action}: The waiting time $Z_{i}$ for $\forall i \in \mathbb{N^+}$.
	\item \textbf{Transition Probability}: The transition probability is given by
	\begin{align} &P[\Omega(R_{i+1})|\Omega(R_{i})] \\
		&=\left\{\begin{aligned} 
			&\tau^{Z_{i+1}+Y_{i+1}}(\Omega(R_{i})),  \quad\quad\: \text{if}\,\, \Omega(R_{i+1}) = 1-q^{(Y_{i+1})}\\
			&1-\tau^{Z_{i+1}+Y_{i+1}}(\Omega(R_{i})), \;\; \text{if}\,\, \Omega(R_{i+1}) = p^{(Y_{i+1})}.
		\end{aligned}\right.\end{align}
	\item \textbf{Cost Function}: The cost function is defined by
	\begin{equation}
		c(\omega,Z_{i+1},\beta) =  \mathbb{E}\left[\sum_{k=0}^{Z_{i+1} +Y_{i+1}-1} (H(\tau^{k}(\omega))-\beta)\right],
	\end{equation}
	where $\Omega(R_i) = \omega$.
\end{itemize}

As the state space is finite and the optimality of the MDP can be achieved by a unichain policy, the optimal policy $\phi^*_{\beta}$ is determined by the optimality equation as follows:
\begin{equation} \label{b-11}
	\begin{aligned}
		&g(\beta) + V(\omega,\beta)= \\
		&  \mathop{\inf}_{Z_{i+1} \in \mathbb{N}}  \Bigg\{c(\omega,Z_{i+1},\beta) + \mathbb{E}\Bigg[\tau^{Z_{i+1}+Y_{i+1}}(\omega)V(1-q^{(Y_{i+1})},\beta)\\
	&+(1-\tau^{Z_{i+1}+Y_{i+1}}(\omega))V(p^{(Y_{i+1})},\beta) \Bigg]\Bigg\},
	\end{aligned}
\end{equation}
where $g(\beta)$ is the average cost achieved by the optimal policy, $V(\omega,\beta)$ is the value function.

This infinite-horizon MDP problem can be solved by using relative value iteration (RVI) \cite[Section 6.5.2]{krishnamurthy2016partially}. To adapt to the RVI algorithm, we rewrite the optimality equation \eqref{b-11} in the following form that is anchored at state $p$. Define $\widetilde{V}(\omega,\beta) = V(\omega,\beta) - V(p,\beta)$, so $\widetilde{V}(p,\beta) = 0$. Then \eqref{b-11} can be rewritten relative to state $p$ as
\begin{equation} \label{b-12}
	\begin{aligned}
		g(\beta) + \widetilde{V}(\omega,\beta) = \mathop{\inf}_{Z_{i+1} \in \mathbb{N}}  \left\{c(\omega,Z_{i+1},\beta) + r(\omega,Z_{i+1},\beta)\right\},
	\end{aligned}
\end{equation}
\begin{equation} \label{b-13}
	g(\beta) = \mathop{\inf}_{Z_{i+1} \in \mathbb{N}}  \left\{c(p,Z_{i+1},\beta) + r(p,Z_{i+1},\beta)\right\},
\end{equation}
where 
\begin{equation} \label{b-14}
	\begin{aligned}
		r(\omega,Z_{i+1},\beta)=&\mathbb{E}\Bigg[\tau^{Z_{i+1}+Y_{i+1}}(\omega)\widetilde{V}(1-q^{(Y_{i+1})},\beta)\\
		&+(1-\tau^{Z_{i+1}+Y_{i+1}}(\omega))\widetilde{V}(p^{(Y_{i+1})},\beta) \Bigg] .
	\end{aligned}
\end{equation}
It is worthwhile to note that we assume that the state $p$ belongs to the state space here, if not the anchor point can be changed to other state in the state space.

Next, we proof that the optimal policy $\pi^*$ for problem \eqref{eq3-5} and the optimal policy $\phi^*_{\bar{p}_{\mathrm{opt}}}$ for problem $g(\bar{p}_{\mathrm{opt}})$ in \eqref{b-8} are identical, \emph{i.e.}, $\pi^*=\phi^*_{\bar{p}_{\mathrm{opt}}}$.
By similarly applying Dinkelbach’s method for nonlinear fractional programming as in \cite{dinkelbach1967nonlinear} and \cite[lemma 2]{DBLP:journals/tit/SunPU20}, we get two assertions in the following lemma: 
\begin{lemma} \label{lemma2}
	There are two assertions that hold true as follows:
	\begin{enumerate}
		\item $\bar{p}_{\mathrm{opt}} \lesseqgtr\beta$ if and only if $g(\beta) \lesseqgtr 0$.
		\item  $\bar{p}_{\mathrm{opt}}$ is the unique root of $g(\beta)=0$. 
	\end{enumerate}
\end{lemma}
\begin{proof}
	Please see Appendix \ref{D} for detaied proof.
\end{proof}
Consequently, we have $g(\bar{p}_{\mathrm{opt}}) = 0$. And there exists an optimal policy $\phi^*_{\bar{p}_{\mathrm{opt}}}=(Z_1,Z_2,\cdots)$ satisfying that
\begin{equation} \label{b-15}
	\begin{aligned}
		&\mathop{\lim}_{n \rightarrow \infty} \frac{1}{n} \sum_{i=0}^{n-1} \Bigg\{ \mathbb{E}\left[\sum_{k=0}^{Z_{i+1} +Y_{i+1}-1} H(\tau^{k}(\omega))\right] \\ 
		&\qquad\qquad\qquad\qquad\qquad\qquad -\bar{p}_{\mathrm{opt}} \mathbb{E}[Z_{i+1} + Y_{i+1}] \Bigg\} = 0,
	\end{aligned}
\end{equation}
which implies that for policy $\phi^*_{\bar{p}_{\mathrm{opt}}}$, 
\begin{equation}  \label{b-16}
	\mathop{\lim}_{n \rightarrow \infty} \frac{\sum_{i=0}^{n-1} \mathbb{E}\left[\sum_{k=0}^{Z_{i+1} +Y_{i+1}-1} H(\tau^{k}(\omega))\right] }{\sum_{i=0}^{n-1} \mathbb{E}[Z_{i+1} + Y_{i+1}] } =\bar{p}_{\mathrm{opt}}.
\end{equation}
Because \eqref{b-16} is equal to \eqref{b-7}, the policy  $\phi^*_{\bar{p}_{\mathrm{opt}}}$ is also the optimal policy of problem \eqref{eq3-5}, \emph{i.e.}, $\pi^*=\phi^*_{\bar{p}_{\mathrm{opt}}}$.

Therefore, the optimality equation of the SMDP problem \eqref{b-5} can be converted into the optimality equation of the MDP problem $g(\bar{p}_{\mathrm{opt}})$ in \eqref{b-8}, given by
\begin{equation} \label{eqb-20}
	\begin{aligned}
		&g(\bar{p}_{\mathrm{opt}}) + \widetilde{V}(\omega,\bar{p}_{\mathrm{opt}}) \\
		&= \mathop{\inf}_{Z_{i+1} \in \mathbb{N}}  \left\{c(\omega,Z_{i+1},\bar{p}_{\mathrm{opt}}) + r(\omega,Z_{i+1},\bar{p}_{\mathrm{opt}})\right\},
	\end{aligned}
\end{equation}
\begin{equation} \label{eqb-21}
	g(\bar{p}_{\mathrm{opt}}) = \mathop{\inf}_{Z_{i+1} \in \mathbb{N}}  \left\{c(p,Z_{i+1},\bar{p}_{\mathrm{opt}}) + r(p,Z_{i+1},\bar{p}_{\mathrm{opt}})\right\},
\end{equation}
where 
\begin{equation}
	c(\omega,Z_{i+1},\bar{p}_{\mathrm{opt}}) =  \mathbb{E}\left[\sum_{k=0}^{Z_{i+1} +Y_{i+1}-1} (H(\tau^{k}(\omega))-\bar{p}_{\mathrm{opt}})\right],
\end{equation}
\begin{equation}
	\begin{aligned}
		r(\omega,Z_{i+1},\bar{p}_{\mathrm{opt}})=&\mathbb{E}\Bigg[\tau^{Z_{i+1}+Y_{i+1}}(\omega)\widetilde{V}(1-q^{(Y_{i+1})},\bar{p}_{\mathrm{opt}})\\
		&+(1-\tau^{Z_{i+1}+Y_{i+1}}(\omega))\widetilde{V}(p^{(Y_{i+1})},\bar{p}_{\mathrm{opt}}) \Bigg] ,
	\end{aligned}
\end{equation}
for all $\omega \in \{ p^{(n)},\, 1-q^{(n)} \}$, $n \in \mathbb{N}$.

This completes the proof.

\section{The proof of Theorem \ref{theorem2}} \label{B}
The sampling problem \eqref{eq3-7} is an infinite-horizon average cost SMDP \cite[Chapter 5.6]{dynamic2012}. The components of this problem are the same as the SMDP in the proof of Theorem \ref{theorem1} in Appendix \ref{A} with one exception. In this case, from the assumption in \eqref{eq3-6}, the transition probability is given as
\begin{align} P[\Omega(R_{i+1})|\Omega(R_{i})]  \label{eqB-1}
	=\left\{\begin{aligned} 
		&\omega^{*},  \quad\quad\: \text{if}\,\, \Omega(R_{i+1}) = 1\\
		&1-\omega^{*}, \;\; \text{if}\,\, \Omega(R_{i+1}) = 0.
	\end{aligned}\right.\end{align}

The average-cost SMDP \eqref{eq3-7} can be solved by using dynamic programming \cite{dynamic2012}. If $\Omega(R_i) = \omega$, then the optimal decision $Z_{i+1}$ for the $(i+1)$-th packet satisfies the following Bellman optimality equation
\begin{equation} 
	\begin{aligned} \label{eqB-2}
		V_{\psi}(\omega) = &\mathop{\inf}_{Z_{i+1} \in \mathbb{N}} \Biggl\{ \mathbb{E}\left[\sum_{k=0}^{Z_{i+1} +Y_{i+1}-1} (H(\tau^{k}(\omega))-\bar{p}_{\mathrm{nopt}})\right] \Biggr\}\\
		&+  \mathbb{E}\Big[\omega^{*}V_{\psi}(1-q^{(Y_{i+1})}) + (1-\omega^{*})V_{\psi}(p^{(Y_{i+1})})\Big] ,
	\end{aligned}
\end{equation}
for all $\omega \in \{ p^{(n)},\, 1-q^{(n)} \}$, $n \in \mathbb{N}$, where $V_{\psi}(\omega)$ is the relative value function associated with the average cost problem \eqref{eq3-7}. 

Let $\bar{Z} = Z(\omega, \bar{p}_\mathrm{nopt})$ be an optimal solution to \eqref{eqB-2}, which means that given $\Omega(R_i) = \omega$, then the optimal waiting time $Z_{i+1}$ for sending the $(i+1)$-th packet is $\bar{Z}$. Because the terms 
\begin{equation*}
	\mathbb{E}\Big[\omega^{*}V(1-q^{(Y_{i+1})}) + (1-\omega^{*})V(p^{(Y_{i+1})})\Big]
\end{equation*}
do not depend on the waiting time $\bar{Z}$, we can reformulate the optimization problem as
\begin{equation} \label{eqB-5}
	\mathop{\inf}_{\bar{Z} \in \mathbb{N}}  \mathbb{E}\left[\sum_{k=0}^{\bar{Z} +Y_{i+1}-1} (H(\tau^{k}(\omega))-\bar{p}_{\mathrm{nopt}})\right] .
\end{equation}
Based on \eqref{eqB-5}, we can derive that $\bar{Z} = 0$ if
\begin{equation} \label{eqB-6}
	\begin{aligned}
		&\mathop{\inf}_{\bar{Z} \in \mathbb{N^+}}  \mathbb{E}\left[\sum_{k=0}^{\bar{Z} +Y_{i+1}-1}(H(\tau^{k}(\omega))-\bar{p}_{\mathrm{nopt}})\right]  \\
		&\geq \mathbb{E} \left[\sum_{k=0}^{Y_{i+1}-1}  (H(\tau^{k}(\omega))-\bar{p}_{\mathrm{nopt}})\right] .
	\end{aligned}
\end{equation}
After some rearrangement and elimination, the inequality \eqref{eqB-6} can also be expressed as
\begin{equation} \label{eqB-7}
	\mathop{\inf}_{\bar{Z} \in \mathbb{N^+}}  \mathbb{E}\left[\sum_{k=0}^{\bar{Z} -1}[H(\tau^{k+Y_{i+1}}(\omega))-\bar{p}_{\mathrm{nopt}}]\right]  \geq 0.
\end{equation}
Because the left-hand side of \eqref{eqB-7} are the infimum of strictly increasing and linear functions of $\bar{p}_{\mathrm{nopt}}$, The inequality \eqref{eqB-7} holds if and only if
\begin{equation} \label{eqB-8}
	\mathop{\inf}_{\bar{Z} \in \mathbb{N^+}} \Biggl\{ \frac{1}{\bar{Z}} \mathbb{E}\left[\sum_{k=0}^{\bar{Z} -1}[H(\tau^{k+Y_{i+1}}(\omega))]\right] \Biggr\} \geq \bar{p}_{\mathrm{nopt}}.
\end{equation}
The left-hand side of \eqref{eqB-8} has the same expression as the index function $\eta(\omega)$ given by \eqref{eq3-8}. Similarly, $\bar{Z} = 1$, if $\bar{Z} \neq 0$ and
\begin{equation} \label{eqB-9}
	\mathop{\inf}_{\bar{Z} \in \{2,3,\cdots\}}  \mathbb{E}\left[\sum_{k=0}^{\bar{Z} -1}[H(\tau^{k+Y_{i+1}}(\omega))-\bar{p}_{\mathrm{nopt}}]\right] \geq 0,
\end{equation}
which is equivalent to
\begin{equation} \label{eqB-10}
	\eta(\tau(\omega)) \geq \bar{p}_{\mathrm{nopt}}.
\end{equation}

By repetition of the same steps as in \eqref{eqB-8}-\eqref{eqB-9}, we can obtain $\bar{Z} = k$ is optimal, if $\bar{Z} \neq 0,1,\cdots,k-1$ and
\begin{equation} \label{eqB-11}
	\eta(\tau^k(\omega)) \geq \bar{p}_{\mathrm{nopt}}.
\end{equation}
Therefore, the optimal waiting time $Z_{i+1} = \bar{Z}$ satisfies
\begin{equation} \label{eqB-12}
	\bar{Z} = Z(\omega,\bar{p}_{\mathrm{nopt}}) = \min \{k \in \mathbb{N}:\eta(\tau^k(\omega)) \geq \bar{p}_{\mathrm{nopt}} \}.
\end{equation}

Next, we compute the optimal objective value $\bar{p}_{\mathrm{nopt}}$. Rewrite $\Omega(R_i)$ as
$	\Omega(R_i(\beta_{\psi})) =\omega= \tau^{Y_i} (\omega_0)$,
where $\omega_0 = 1-q$ or $\omega_0 = p$.
And let $Z_{i}$ denote the waiting time of the $i$-th submitted information at the receiver. Because $R_{i+1}(\beta_{\psi}) = R_i(\beta_{\psi}) + Z_{i+1}(\beta_{\psi}) + Y_{i+1}$, then we have 
\begin{equation} \label{eq3-22}
	\begin{aligned}
		&\mathbb{E} \left[	\sum_{t=R_i(\beta_{\psi})}^{R_{i+1}(\beta_{\psi})-1} H(\Omega(t))\right] - \beta_{\psi} \mathbb{E} \left[ R_{i+1}(\beta_{\psi})-R_i(\beta_{\psi}) \right] \\
		&= \mathbb{E} \left[ \sum_{k=0}^{Z_{i+1}(\beta_{\psi})-1} H(\tau^{k+Y_{i+1}+Y_{i}}(\omega_0)) \right] \\
		&+ \mathbb{E} \left[ \sum_{k=0}^{Y_{i+1}-1} H(\tau^{k+Y_{i}}(\omega_0)) \right] - \beta_{\psi} \mathbb{E}\left[ Z_{i+1}(\beta_{\psi}) + Y_{i+1} \right],
	\end{aligned}
\end{equation}
which implies that $\bar{p}_{\mathrm{nopt}} = \beta_{\psi}$ is the root of \eqref{eq3-5}, if the following equation holds
\begin{equation} \label{eq3-23}
	\begin{aligned}
		&\mathbb{E} \left[ \sum_{k=0}^{Z_{i+1}(\beta_{\psi}) -1} H(\tau^{k+Y_{i+1}+Y_{i}}(\omega_0)) \right] +\\
		& \mathbb{E} \left[ \sum_{k=0}^{Y_{i+1}-1} H(\tau^{k+Y_{i}}(\omega_0)) \right] - \bar{p}_{\mathrm{nopt}} \mathbb{E}\left[ Z_{i+1}(\beta_{\psi}) + Y_{i+1} \right] = 0.
	\end{aligned}
\end{equation}
From \eqref{eqB-2}, we can get that
\begin{equation}  \label{eq3-24}
	\begin{aligned}
		&V_{\psi} (\tau^{Y_i}(\omega_0)) \\
		&= \mathbb{E} \left[ \sum_{k=0}^{Y_{i+1}-1} [H(\tau^{k+Y_i}(\omega_0)) - \bar{p}_{\mathrm{nopt}}] \right] \\
		&+ \mathbb{E} \left[ \sum_{k=0}^{Z_{i+1}(\beta_{\psi}) -1} [H(\tau^{k+Y_{i+1}+Y_i}{(\omega_0)}) - \bar{p}_{\mathrm{nopt}}] \right] \\
		&+ \mathbb{E} \left[ \omega^* V_{\psi}(\tau^{Y_{i+1}}(1-q)) +  (1-\omega^*)V_{\psi}(\tau^{Y_{i+1}}(p)) \right],
	\end{aligned}
\end{equation}
According to the law of iterated expectations, if we take expectation over $\omega_0$ on both sides of \eqref{eq3-24}, the following equation is obtained
	\begin{align} \label{c-18}
		&\mathbb{E} \left[ V_{\psi} (\tau^{Y_i}(\omega_0)) \right] \notag \\
		&= \mathbb{E} \left[ \sum_{k=0}^{Y_{i+1}-1} [H(\tau^{k+Y_i}(\omega_0)) - \bar{p}_{\mathrm{nopt}}] \right] \notag\\
		&+ \mathbb{E} \left[ \sum_{k=0}^{Z_{i+1}(\beta_{\psi}) -1} [H(\tau^{k+Y_{i+1}+Y_i}(\omega_0)) - \bar{p}_{\mathrm{nopt}}] \right] \notag\\
		&+ \mathbb{E} \left[ V_{\psi}(\tau^{Y_{i+1}}(\omega_0)) \right].
	\end{align}
Because the remote source is a finite-state ergodic Markov chain with a unique stationary distribution, we have that $\mathbb{E} \left[ V_{\psi} (\tau^{Y_i}(\omega_0)) \right] = \mathbb{E} \left[ V_{\psi} (\tau^{Y_{i+1}}(\omega_0)) \right]$. Consequently, \eqref{c-18} is transformed into
\begin{equation}
	\begin{aligned}
		&\mathbb{E} \left[ \sum_{k=0}^{Y_{i+1}-1} [H(\tau^{k+Y_i}(\omega_0)) - \bar{p}_{\mathrm{nopt}}] \right] \\
		&+ \mathbb{E} \left[ \sum_{k=0}^{Z_{i+1}(\beta_{\psi}) -1} [H(\tau^{k+Y_{i+1}+Y_i}(\omega_0)) - \bar{p}_{\mathrm{nopt}}] \right] =0,
	\end{aligned}
\end{equation}
which is equivalent to \eqref{eq3-23}.

Finally, we prove that the root of \eqref{eq3-5} is unique. According to the optimal index policy, the right-hand side of \eqref{eq3-22} can be expressed as
\begin{equation} \label{c-20}
	\begin{aligned}
		&\mathbb{E} \left[ \sum_{k=0}^{Z_{i+1}(\beta_{\psi})-1} [H(\tau^{k+Y_{i+1}+Y_{i}}(\omega_0)) - \beta_{\psi}]\right] \\
		&+ \mathbb{E} \left[ \sum_{k=0}^{Y_{i+1}-1} [H(\tau^{k+Y_{i}}(\omega_0)) -\beta_{\psi}] \right] \\
		&= \mathop{\inf}_{\mu \in \mathbb{N}} \mathbb{E} \left[ \sum_{k=0}^{\mu -1} [H(\tau^{k+Y_{i+1}+Y_{i}}(\omega_0)) - \beta_{\psi}] \right] \\
		&+ \mathbb{E} \left[ \sum_{k=0}^{Y_{i+1}-1} [H(\tau^{k+Y_{i}}(\omega_0)) -\beta_{\psi}] \right].
	\end{aligned}
\end{equation}
Since the first term is the pointwise infimum of the linear decreasing functions of $\beta_{\psi}$ and the second term is a linear decreasing function of $\beta_{\psi}$, the right-hand side of \eqref{c-20} is a continuous, concave, and strictly decreasing function of $\beta_{\psi}$. Thus the first term of the right-hand side of \eqref{c-20} satisfies
\begin{equation}
	\mathop{\lim}_{\beta_{\psi} \rightarrow -\infty} \mathop{\inf}_{\mu \in \mathbb{N}} \mathbb{E} \left[ \sum_{k=0}^{\mu -1} [H(\tau^{k+Y_{i+1}+Y_{i}}(\omega_0)) - \beta_{\psi}] \right] = \infty,
\end{equation}
and 
\begin{equation}
	\mathop{\lim}_{\beta_{\psi} \rightarrow \infty} \mathop{\inf}_{\mu \in \mathbb{N}} \mathbb{E} \left[ \sum_{k=0}^{\mu -1} [H(\tau^{k+Y_{i+1}+Y_{i}}(\omega_0)) - \beta_{\psi}] \right] = -\infty,
\end{equation}
and the second term of the right-hand side of \eqref{c-20} satisfies
\begin{equation}
	\mathop{\lim}_{\beta_{\psi} \rightarrow -\infty} \mathbb{E} \left[ \sum_{k=0}^{Y_{i+1}-1} [H(\tau^{k+Y_{i}}(\omega_0)) -\beta_{\psi}] \right] = \infty,
\end{equation}
and
\begin{equation}
	\mathop{\lim}_{\beta_{\psi} \rightarrow \infty} \mathbb{E} \left[ \sum_{k=0}^{Y_{i+1}-1} [H(\tau^{k+Y_{i}}(\omega_0)) -\beta_{\psi}] \right] = -\infty.
\end{equation}
 Therefore, we can obtain that \eqref{eq3-5} has a unique root. This completes the proof.

\section{Proof of Lemma \ref{lemma2}} \label{D}
\subsection{Proof of Part ($i$)}
We prove that $\bar{p}_{\mathrm{opt}} \lesseqgtr\beta$ if and only if $g(\beta) \lesseqgtr 0$.

1) $\bar{p}_{\mathrm{opt}} \leq \beta \Leftrightarrow g(\beta) \leq 0$

If $\bar{p}_{\mathrm{opt}} \leq \beta$, there exists a policy $\pi=(Z_1,Z_2,\cdots) \in \Pi$ that is feasible for both \eqref{b-7} and \eqref{b-8} such that
\begin{equation} \label{d-1}
	\mathop{\lim}_{n \rightarrow \infty} \frac{\sum_{i=0}^{n-1} \mathbb{E}_{\pi}\left[\sum_{k=0}^{Z_{i+1} +Y_{i+1}-1} H(\tau^{k}(\omega))\right] }{\sum_{i=0}^{n-1} \mathbb{E}_{\pi}[Z_{i+1} + Y_{i+1}] } \leq \beta,
\end{equation}
which is equivalent to
\begin{equation}\label{d-2}
	\mathop{\lim}_{n \rightarrow \infty} \frac{\frac{1}{n}\sum_{i=0}^{n-1} \mathbb{E}_{\pi}\left[\sum_{k=0}^{Z_{i+1} +Y_{i+1}-1} [H(\tau^{k}(\omega))-\beta]\right] }{\frac{1}{n}\sum_{i=0}^{n-1} \mathbb{E}_{\pi}[Z_{i+1} + Y_{i+1}] } \leq 0.
\end{equation}
Since the inter-sampling times ($Y_{i+1}+Z_{i+1}$) are regenerative \cite[Section 6.1]{haas2006stochastic}, $Y_{i+1}>0$, and $0 \leq Z_{i+1}<\infty$, the limit $\mathop{\lim}_{n \rightarrow \infty} \frac{1}{n}\sum_{i=0}^{n-1} \mathbb{E}[Z_{i+1} + Y_{i+1}]$ always exists and positive, which means that
\begin{equation}\label{d-3}
	0<\mathop{\lim}_{n \rightarrow \infty} \frac{1}{n}\sum_{i=0}^{n-1} \mathbb{E}[Z_{i+1} + Y_{i+1}]<\infty.
\end{equation}
for any policy. Thus we have
\begin{equation}\label{d-4}
	\mathop{\lim}_{n \rightarrow \infty} \frac{1}{n}\sum_{i=0}^{n-1} \mathbb{E}\left[\sum_{k=0}^{Z_{i+1} +Y_{i+1}-1} [H(\tau^{k}(\omega))-\beta]\right] \leq 0,
\end{equation}
which implies that $g(\beta) \leq 0$, as $g(\beta)$ is the infimum value.

On the reverse direction, if $g(\beta) \leq 0$, we can know that there exists a policy $\pi=(Z_1,Z_2,\cdots) \in \Pi$ that is feasible for both \eqref{b-7} and \eqref{b-8} satisfying \eqref{d-4}. Because \eqref{d-3} always holds, we can obtain that \eqref{d-2} and \eqref{d-1} holds. Therefore, $\bar{p}_{\mathrm{opt}} \leq \beta$.

2) $\bar{p}_{\mathrm{opt}} > \beta \Leftrightarrow g(\beta) > 0$

If $\bar{p}_{\mathrm{opt}} > \beta$, for any policy $\pi = (Z_1,Z_2,\cdots) \in \Pi$ that is feasible for both \eqref{b-7} and \eqref{b-8} we have that
\begin{equation} \label{d-5}
	\mathop{\lim}_{n \rightarrow \infty} \frac{\sum_{i=0}^{n-1} \mathbb{E}\left[\sum_{k=0}^{Z_{i+1} +Y_{i+1}-1} H(\tau^{k}(\omega))\right] }{\sum_{i=0}^{n-1} \mathbb{E}[Z_{i+1} + Y_{i+1}] } > \beta.
\end{equation}
With \eqref{d-3} in hand, we derive that for any policy
\begin{equation}\label{d-6}
	\mathop{\lim}_{n \rightarrow \infty} \frac{1}{n}\sum_{i=0}^{n-1} \mathbb{E}\left[\sum_{k=0}^{Z_{i+1} +Y_{i+1}-1} [H(\tau^{k}(\omega))-\beta]\right] > 0,
\end{equation}
which implies that $g(\beta) > 0$.

On the contrary, if the infimum value $g(\beta)>0$, then \eqref{d-6} holds for any policy. Similar to the proof of $\bar{p}_{\mathrm{opt}} \leq \beta \Leftrightarrow g(\beta) \leq 0$, we can finally derive that $\bar{p}_{\mathrm{opt}} > \beta$.

\subsection{Proof of Part ($ii$)}
From \eqref{b-8}, we can see that $g(\beta)$ is the pointwise infimum of the linear decreasing functions of $\beta$. Consequently, the root of $g(\beta) = 0$ is unique. Combining with the assertion in Part ($i$), we get that $\bar{p}_{\mathrm{opt}}$ is the unique root of $g(\beta)=0$. 

This completes the proof.

\clearpage

\end{document}